\documentclass[letterpaper, 10 pt, conference]{ieeeconf} 
\IEEEoverridecommandlockouts                             
\usepackage{xspace}
\usepackage{times}
\usepackage{amsmath,amssymb,amsfonts}
\usepackage{verbatim}
\usepackage{graphicx}
\usepackage{subfigure}
\usepackage{multirow}
\usepackage{psfrag,graphicx,epsfig,epsf}
\usepackage{graphics}
\usepackage{latexsym}
\usepackage[ruled,linesnumbered, noend]{algorithm2e}
\usepackage{fancyhdr}
\usepackage{fullpage}
\usepackage{wrapfig}
\usepackage{subfigure}
\usepackage{setspace}
\usepackage[bookmarks=true]{hyperref}
\usepackage{multicol}
\usepackage{color}

\DeclareMathOperator*{\argmin}{arg\,min}

\usepackage{amsmath,amsfonts,amssymb,mathrsfs}
\usepackage{centernot}

\newtheorem{definition}{\bf Definition}
\newtheorem{theorem}{\bf Theorem}
\newtheorem{lemma}{\bf Lemma}

\title{\Large \bf Scalable Asymptotically-Optimal 
Multi-Robot Motion Planning}

\author{Andrew Dobson \and Kiril Solovey \and Rahul Shome \and Dan
  Halperin \and Kostas E. Bekris% 
\thanks{A. Dobson, R. Shome and K. E. Bekris are with the Computer
  Science Dept. of Rutgers University, NJ, USA, {\tt \{chuples.dobson,
    rahul.shome, kostas.bekris\}@cs.rutgers.edu}. Their work is
  supported by NSF IIS 1451737 and NSF CCF 1330789.}%
\thanks{K. Solovey and D. Halperin are with the Computer Science
  Dept. of Tel Aviv University, Israel,  {\tt
    \{kirilsol,danha\}@post.tau.ac.il}. Their work has been supported in part
  by the Israel Science Foundation (grant no.~825/15) and by the
  Blavatnik Computer Science Research Fund. K. S. has also been
  supported by the Clore Israel Foundation.}%
}

\begin{document}

\maketitle

%%%%%%%%%%%%%%%%%%%%%%%%%%%%%%%%%%%%
%% Nick Saving space
%%%%%%%%%%%%%%%%%%%%%%%%%%%%%%%%%%%%
% Space between figure and caption
\setlength{\abovecaptionskip}{-2.5pt}
\setlength{\belowcaptionskip}{-6pt}
% Space between text and figs
\setlength{\dbltextfloatsep}{2pt plus 1.0pt minus 1.0pt}
\setlength{\textfloatsep}{2pt plus 1.0pt minus 1.0pt}
\setlength{\intextsep}{2pt plus 1.0pt minus 1.0pt}
% Space between equations and text
\setlength{\belowdisplayskip}{0pt} \setlength{\belowdisplayshortskip}{2pt}
\setlength{\abovedisplayskip}{0pt} \setlength{\abovedisplayshortskip}{2pt}

\newcommand{\mam}{$\mathcal{G}_{\tt MAM}$}
\newcommand{\pr}{\ensuremath{\mathbb{P}}}

%%% Mathematical Definitions
\newcommand{\reals}{\mathbb{R}}
\newcommand{\integers}{\mathbb{Z}}

%%% Definitions for Workspace, Objects, Manipulator, and Tasks
\newcommand{\Wspace}{\mathbb{W}}
\newcommand{\Sspace}{\mathscr{S}}
\newcommand{\Robots}{R}
\newcommand{\Manip}{\mathbb{A}}
\newcommand{\Obstacles}{\mathscr{Z}}
\newcommand{\obst}{\mathbb{Z}}
\newcommand{\cspace}{\ensuremath{\mathbb{C}}}
\newcommand{\ccross}{\ensuremath{\mathbb{C}}}
\newcommand{\ccrossfree}{\ensuremath{\mathbb{C}^{\textup{f}}}}
\newcommand{\cfree}{\ensuremath{\cspace^{\textup{f}}}}
\newcommand{\cinv}{\cspace^{\textup{o}}}
\newcommand{\cbound}{\cspace_{\cap}}
\newcommand{\cstable}{\cspace_{s}}
\newcommand{\cgrasp}{\cspace_{G}}

\newcommand{\oracle}{\mathbb{O}_d}

%% Stuffs for algorithmics
\newcommand{\qrand}{\ensuremath{Q^{\textup{rand}}}}
\newcommand{\vnear}{\ensuremath{V^{\textup{near}}}}
\newcommand{\vnew}{\ensuremath{V^{\textup{new}}}}
\newcommand{\vlast}{\ensuremath{V^{\textup{last}}}}
\newcommand{\vparent}{\ensuremath{V^{\textup{best}}}}

%% Definitions for Object stuff
\newcommand{\Pspace}{\mathbb{P}}
\newcommand{\pose}{p}

%% Definitions for Manipulator stuff
\newcommand{\Qspace}{\mathbb{Q}}
\newcommand{\GeomManip}{\mathbb{WM}}

\newcommand{\rad}{\ensuremath{r(n)}}
\newcommand{\radstar}{\ensuremath{r^*(n)}}
\newcommand{\radi}{\ensuremath{r_i(n)}}
\newcommand{\radj}{\ensuremath{r_j(n)}}
\newcommand{\crossrad}{\ensuremath{r_R(n)}}
\newcommand{\crossradstar}{\ensuremath{r^*_R(n)}}
\newcommand{\impcrossrad}{\ensuremath{\hat r_R(n)}}
\newcommand{\allimpcrossrad}{\ensuremath{\hat r_{R}(n^R)}}
\newcommand{\ki}{\ensuremath{k_i(n)}}
\newcommand{\kj}{\ensuremath{k_j(n)}}

%% Manipulation roadmap definition
\newcommand{\mmgraph}{\ensuremath{\mathbb{G}}}
\newcommand{\mmgimp}{\hat\mmgraph}
\newcommand{\mmgexp}{\mmgraph}
\newcommand{\graph}{\ensuremath{\mathbb{G}}}
\newcommand{\aograph}{\ensuremath{\mathbb{G}^{AO}}}
\newcommand{\tree}{\ensuremath{\mathbb{T}}}
\newcommand{\mmnodes}{\mathbb{\hat V}}
\newcommand{\mmedges}{\mathbb{\hat E}}
\newcommand{\mmnodestpprm}{\mathbb{V}_{\chi_i}}
\newcommand{\mmedgestpprm}{\mathbb{E}_{\chi_i}}
\newcommand{\mmnode}{\mathbb{\hat v}}
\newcommand{\mmedge}{\mathbb{\hat e}}
\newcommand{\nodes}{\mathbb{V}}
\newcommand{\node}{\mathbb{v}}
\newcommand{\edges}{\mathbb{E}}
\newcommand{\edge}{\mathbb{e}}
\newcommand{\prmstar}{\ensuremath{ {\tt PRM^*} }}
\newcommand{\sprmstar}{Soft-\ensuremath{ {\tt PRM} }}
\newcommand{\irs}{\ensuremath{ {\tt IRS} }}
\newcommand{\spars}{{\tt SPARS}}
\newcommand{\drrt}{\ensuremath{{\tt dRRT}}}
\newcommand{\drrtstar}{\ensuremath{{\tt dRRT^*}}}

\newcommand{\sig}{{\tt SIG}}
\newcommand{\local}{\mathbb{L}}
\newcommand{\rmaps}{\ensuremath{\mathfrak{R}}}

\newcommand{\prm}{{\tt PRM}}
\newcommand{\mmprm}{\ensuremath{\text{Random-}{\tt MMP}}}
\newcommand{\kprmstar}{{\tt k-PRM$^*$}}
\newcommand{\rrt}{\ensuremath{{\tt RRT}}}
\newcommand{\rrtdrain}{{\tt RRT-Drain}}
\newcommand{\rrg}{{\tt RRG}}
\newcommand{\est}{{\tt EST}}
\newcommand{\rrtstar}{\ensuremath{\tt RRT^{\text *}}}
\newcommand{\astar}{{\ensuremath{\tt A^{\text *}}}}
\newcommand{\mstar}{{\tt M^{\text *}}}
\newcommand{\opens}{P_{Heap}}

\newcommand{\bvp}{{\tt BVP}}
\newcommand{\alg}{{\tt ALG}}
\newcommand{\fixed}{{\tt Fixed}-$\alpha$-\rdg}

\newcommand{\config}{C}

\newcommand{\cost}{\textup{cost}}

\newenvironment{myitem}{\begin{list}{$\bullet$}
{\setlength{\itemsep}{-0pt}
\setlength{\topsep}{0pt}
\setlength{\labelwidth}{0pt}
\setlength{\leftmargin}{10pt}
\setlength{\parsep}{-0pt}
\setlength{\itemsep}{0pt}
\setlength{\partopsep}{0pt}}}%
{\end{list}}

\newtheorem{claim}{\bf Claim}

\newcommand*{\qed}{\hfill\ensuremath{\square}}

\newcommand{\kiril}[1]{{\color{blue} \textbf{Kiril:} #1}}
\newcommand{\chups}[1]{{\color{red} \textbf{Chuples:} #1}}
\newcommand{\rahul}[1]{{\color{green} \textbf{Rahul:} #1}}

\newcommand{\T}{\mathcal{T}}

\begin{abstract}
Discovering high-quality paths for multi-robot problems can be
achieved, in principle, through asymptotically-optimal data structures
in the composite space of all robots, such as a sampling-based roadmap
or a tree. The hardness of motion planning, however, which depends
exponentially on the number of robots, renders the explicit
construction of such structures impractical. This work proposes a
scalable, sampling-based planner for coupled multi-robot problems that
provides desirable path-quality guarantees.  The proposed \drrtstar\
is an informed, asymptotically-optimal extension of a prior
method \drrt, which introduced the idea of building roadmaps for each
robot and implicitly searching the tensor product of these structures
in the composite space.  The paper describes the conditions for
convergence to optimal paths in multi-robot problems.  Moreover,
simulated experiments indicate \drrtstar\ converges to high-quality
paths and scales to higher numbers of robots where various
alternatives fail. It can also be used on high-dimensional challenges,
such as planning for robot manipulators.

\end{abstract}

\section{Introduction and Prior Work}
%% We need more references to multi-robot planning work

Many multi-robot planning applications \cite{Wagner:2015bd,
Gravot:2003kh, Gharbi:2009fu} require high-dimensional platforms to
simultaneously move in a shared workspace, where high-quality paths
must be computed quickly as sensing input is updated, as in
Fig. \ref{fig:setup}.  Preprocessing given knowledge of the static
scene can help the online computation of high-quality paths.
Sampling-based roadmaps can help with such high-dimensional challenges
and provide primitives for preprocessing a static
scene \cite{Kavraki1996Probabilistic-R, LaValle2001}.  These methods
converge to optimal solutions given sufficient
density \cite{Karaman2011Sampling-based-}, i.e., at least
$O(n\log(n))$ edges are needed for a roadmap with $n$ vertices, while
near-optimal solutions are achieved after finite computation
time \cite{Dobson:2015_Finite, Pavone:2015fmt}.

%==================

%Such methods quickly produce solutions by leveraging the
%configuration space ($\cspace$-space) abstraction.  While motion
%planning complexity depends exponentially on the problem's
%dimensionality $d$ \cite{kolountzakis1998analysis, HsuKav98},
%sampling-based approaches scale better than alternative methods.
%Primarily, their performance suffers from the existence of narrow
%passages in $\cfree$ \cite{Shi_Amato2014Narrow}.

%==================

\begin{figure}[t]
\centering
\includegraphics[height=1.3in]{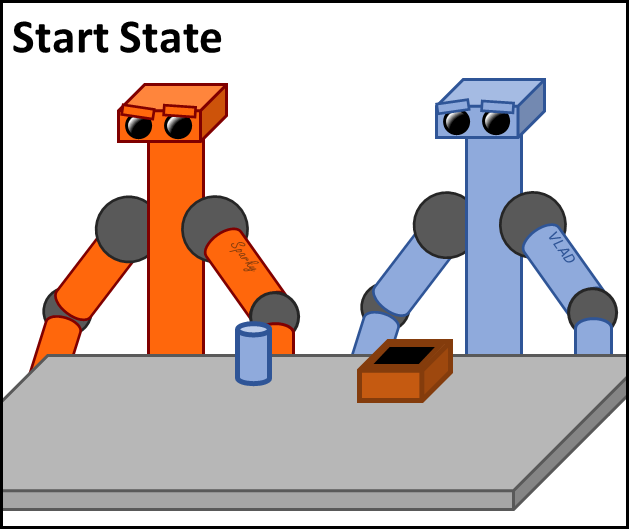}
\includegraphics[height=1.3in]{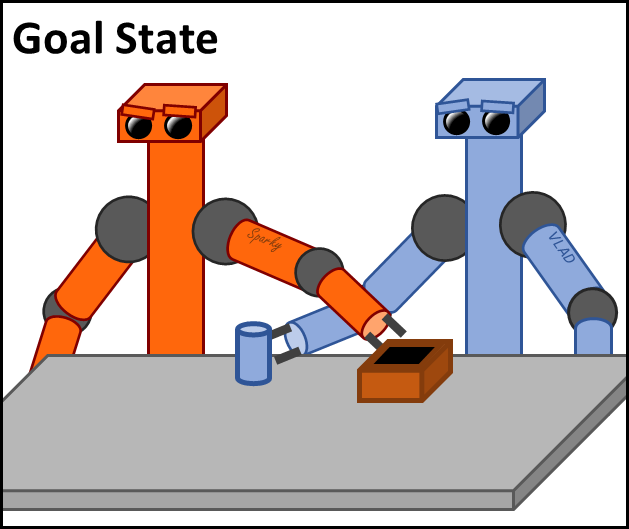}
\caption{Simultaneous planning for multiple high-dimensional systems is a
difficult, motivating challenge for this work.}
\label{fig:setup}
\end{figure}

%==================

Na\"ively constructing a sampling-based roadmap or tree in the robots'
composite configuration space provides asymptotic optimality but does
not scale well. In particular, memory requirements depend
exponentially on the problem's
dimension \cite{Schwartz1983Coordinated-Piano}.  The alternative is
decoupled planing, where paths for robots are computed independently
and then coordinated \cite{Leroy1999Multipath-Coordination}.  These
methods, however, typically lack completeness and optimality
guarantees. Hybrid approaches can achieve optimal decoupling to retain
guarantees
\cite{Berg:2009ve}.  The problem is more complex when the robots
exhibit non-trivial dynamics \cite{Peng2005Kinodynamic-Coord}.
Collision avoidance or control methods can scale to large numbers of
robots, but typically lack global path quality
guarantees \cite{vandenberg2011Reciprocal-Collision,
Tang2015Complete-Multi}.

%==================

The previously proposed \drrt\ approach \cite{SoloveySH16:ijrr} is a
scalable sampling-based approach, which is probabilistically complete.
It searches an implicit tensor product roadmap of roadmaps
constructued for each robot individually \cite{Svestka:1998ud}.  The
current work proposes \drrtstar and shows that it is an efficient
asymptotically optimal variant of the prior method.  Simulations show
the method practically generates high-quality paths while scaling to
complex, high-dimensional problems, where alternatives
fail. \footnote{Additional material is provided in the appendices of 
this paper as well as the accompanying video of the original MRS
submission.}

% \section{Prior Work}
% \input{02_prior_work}

\section{Problem Setup and Notation}

Consider a shared workspace with $R \geq 2$ robots, each operating in
configuration space $\cspace_i$ for $1\leq i\leq R$.  Let $\cfree_i
\subset \cspace_i$ be each robot's free space, where it is free from
collision with the static scene, and $\cinv_i=\cspace_i\setminus
\cfree_i$ is the forbidden space for robot $i$. The composite 
$\cspace$-space $\cspace = \prod^R_{i=1} \cspace_i$ is the Cartesian 
product of each robot's $\cspace$-space.  A composite configuration 
$Q = (q_1,\ldots,q_R) \in \cspace$ is an $R$-tuple of robot 
configurations.  For two distinct robots $i, j$, denote by
$I_i^j(q_j)\subset \cspace_i$ the set of configurations where $i$ and
$j$ collide.  Then, the composite free space $\cfree \subset \cspace$
consists of configurations $Q=(q_1,\ldots,q_R)$ subject to:

\begin{myitem}
\item $q_i \in \cfree_i$ for every $1\leq i\leq R$;
\item $q_i \not\in I_i^j(q_j), q_j \not\in I_j^i(q_i)$ for every 
$1 \leq i < j\leq R$.
\end{myitem}

\noindent Each $Q \in \cfree$ requires robots to not collide with
obstacles, and each pair to not collide with each other. The composite
forbidden space is defined as $\cinv = \cspace \setminus \cfree$.

Given $S,T \in \cfree$, where $S=(s_1,\ldots,s_R),T=(t_1,\ldots,t_R)$,
a \emph{trajectory} $\Sigma:[0,1]\rightarrow \cfree$ is a
continuous curve in $\cfree$, such that $\Sigma(0)=S,\Sigma(1)=T$,
where the $R$ robots move simultaneously. $\Sigma$ is an $R$-tuple 
$(\sigma_1,\ldots,\sigma_R)$ of robot paths such that 
$\sigma_i:[0,1]\rightarrow \cfree_i$.

The objective is to find a trajectory which minimizes a \emph{cost
  function} $c(\cdot)$.  The analysis assumes the cost is the sum of
robot path lengths, i.e., $c(\Sigma)= \sum_{i=1}^R |\sigma_i|$, where
$|\sigma_i|$ denotes the standard \emph{arc length} of $\sigma_i$.
The arguments also work for $\max_{i=1:R} |\sigma_i|$.\footnote{The
  types of distances the arguments hold are more general, but proofs
  for alternative metrics are left as future work.}
Section~\ref{sec:analysis} shows sufficient conditions for
\drrtstar\ to converge to optimal trajectories over the cost function
$c$.

%%% Local Variables:
%%% mode: latex
%%% TeX-master: "../isrr"
%%% End:

\section{Methods for Composite Space Planning}
For a fixed $n \in \mathbb{N}_+$, define for every robot $i$ the PRM
roadmap $\graph_i = (\nodes_i, \edges_i)$ constructed over $\cfree_i$,
such that $|\nodes_i|=n$ with connection radius \rad. Then, $\mmgimp 
= (\mmnodes, \mmedges) = \graph_1\times \ldots \times \graph_R$ is the 
\emph{tensor product roadmap} in space $\ccross$ (for an illustration,
see Figure~\ref{fig:tprm}).  Formally, $\mmnodes 
= \{ ( v_1, v_2, \dots, v_R ), \forall i, v_i \in \nodes_i\}$ is the 
Cartesian product of the nodes from each roadmap $\graph_i$.  For two 
vertices $V =(v_1,\ldots,v_m) \in \mmnodes, V'=(v'_1,\ldots,v'_m) \in 
\mmnodes$ the edge set $\mmedges$ contains edge $(V,V')$ if for every 
$i$ it is that $v_i=v'_i$ or $(v_i,v'_i)\in \edges_i$.\footnote{Notice this slight difference from
  $\drrt$~\cite{SoloveySH16:ijrr} so as to allow edges where some
  robots remain motionless while others move.}

\begin{figure}[t]
\centering
\includegraphics[height=1.5in]{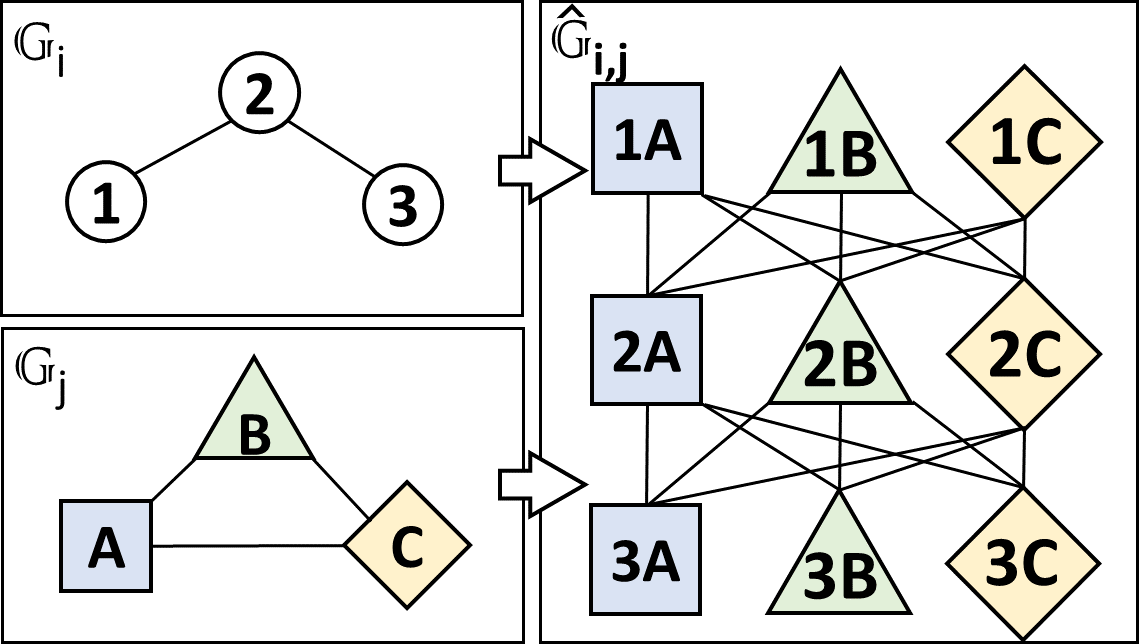}
\caption{An illustration of a two-robot tensor product roadmap $\hat
  \graph_{i,j}$ between roadmaps $\graph_i$ and $\graph_j$. Two nodes
  in the tensor-product roadmap share an edge if all the individual
  robot configurations share an edge in the individual robot
  roadmaps.}
\label{fig:tprm}
\end{figure}

As shown in Algorithm~\ref{algo:drrtstar}, $\drrtstar$ grows a tree
$\tree$ over $\mmgimp$, rooted at the start configuration $S$ and
initializes path $\pi_{\textup{best}}$ (line~1).  The method 
stores the node added each iteration $V$ (Line~2), as part of an 
informed process to guide the expansion of \tree\ towards the goal.  
The method iteratively expands $\tree$ given a time budget (Line~3), 
as detailed by Algorithm~\ref{algo:drrtstar_expand}, storing the newly 
added node $V$ (Line~4).  After expansion, the method traces the path 
which connects the source $S$ with the target $T$ (Line~5).  If such a 
path is found, it is stored in $\pi_{\textup{best}}$ if it improves 
upon the cost of the previous solution (Lines~6,~7).  Finally, the 
best path found $\pi_{\textup{best}}$ is returned (Line~8).

\begin{algorithm}[ht]
\caption{$\drrtstar {\tt(} \mmgimp, S, T {\tt)}$}
\label{algo:drrtstar}
$\pi_{\textup{best}} \gets \emptyset$,
$\tree.{\tt init}(S)$\;
$V \gets S$\;
\While{${\tt time.elapsed}() < {\tt time\_limit}$}
{
    $V \gets {\tt Expand\_\drrtstar(} \mmgimp, \tree, V, T {\tt)}$\;
    $\pi \gets {\tt Trace\_Path(} \tree, S, T {\tt)}$\;
    \If{$\pi \neq \emptyset$ and $cost(\pi) < cost(\pi_{\textup{best}})$ }
        {
            $\pi_{\textup{best}} \gets \pi$\;
        }
}
{\bf return $\pi_{\textnormal{best}}$}
\end{algorithm}

%%% Local Variables:
%%% mode: latex
%%% TeX-master: "../isrr"
%%% End:

The expansion step is given in Alg. \ref{algo:drrtstar_expand}.  The
default initial step of the method is given in Lines~1-4, i.e., when
no $\vlast$ is passed (Line~1), which corresponds to an exploration
step similar to {\tt RRT}: a random sample \qrand\ is generated in
\ccross\ (Line~2), its nearest neighbor \vnear\ in $\tree$ is found
(Line~3) and the oracle function $\oracle(\cdot,\cdot)$ returns the
implicit graph node \vnew\ that is a neighbor of \vnear on the
implicit graph in the direction of $\qrand$ (Line~4). If a $\vlast$,
however, is provided (Line~5)---which happens when the last iteration
managed to generate a node closer to the goal relative to its
parent---then the $\vnew$ is greedily generated so as to be a
neighbor of $\vlast$ in the direction of the goal $T$ (Line~6).

\begin{algorithm}[ht]
\caption{${\tt Expand\_\drrtstar(} \mmgimp, \tree, \vlast, T {\tt)}$}
\label{algo:drrtstar_expand}
\If{$\vlast == NULL$}
{
    $\qrand \gets {\tt Random\_Sample()}$\;
    $\vnear \gets {\tt Nearest\_Neighbor(} \tree, \qrand {\tt)}$\;
    $\vnew \gets {\tt \oracle (} \vnear, \qrand {\tt)}$\;
}
\Else
{
    $\vnew \gets {\tt \oracle (} \vlast, T {\tt)}$\;
}
$N \gets {\tt Adjacent(} \vnew, \mmgimp {\tt )} \cap \nodes_{\tree}$\;
$\vparent \gets \argmin_{v \in N s.t. \local(v, \vnew) \subset \cfree } c(v) + c( \local(v, \vnew )) $\;
\If{$\vparent == NULL$}
{
    ${\bf return}\ NULL$\;
}
\If{ $\vnew \notin \tree$ }
{
    $\tree.{\tt Add\_Vertex(} \vnew {\tt)}$\;
    $\tree.{\tt Add\_Edge(} \vparent, \vnew {\tt)}$\;
}
\Else
{
    $\tree.{\tt Rewire(} \vparent, \vnew {\tt)}$\;    
}
\For{ $v \in N$ }
{
    \If{ $c( \vnew ) + c( \local( \vnew, v ) < c( v )$ and $\local(\vnew, v) \subset \cfree$ }
    {
        $\tree.{\tt Rewire(} \vnew, v {\tt )}$\;
    }
}
\If{ $h(\vnew) < h(\vparent)$ }
{
  ${\bf return}\ \vnew$\;
}
\Else
{
    ${\bf return}\ NULL$\;
}
\end{algorithm}

%%% Local Variables:
%%% mode: latex
%%% TeX-master: "../isrr"
%%% End:

In either case, the method next finds neighbors $N$, which
are adjacent to \vnew\ in $\mmgimp$ and have also been added to
$\tree$ (Line~7).  Among $N$, the best node \vparent is chosen, for 
which the local path $\local(\vparent, \vnew)$ is collision-free and 
that the total path cost to $\vnew$ is minimized (Line~8).  If no such 
parent can be found (Line~9), the expansion fails and no node is 
returned (Line~10).  Then, if $\vnew$ is not in $\tree$, it is added
(Lines~11-13).  Otherwise, if it exists, the tree is rewired so as to 
contain edge $(\vparent, \vnew)$, and the cost of the $\vnew$'s
sub-tree (if any) is updated (Lines~14,~15).  Then, for all nodes in
$N$ (Line~16), the method tests $\tree$ should be rewired through 
$\vnew$ to reach this neighbor.  Given that $\local(\vnew, v)$ is 
collision-free and is of lower cost than the existing path to $v$ 
(Line~17), the tree is rewired to make $\vnew$ be the parent of $v$ 
(line~18).

Finally, if in this iteration the heuristic value of $\vnew$ is 
lower than its parent node $\vparent$ (line 19), the method returns 
$\vnew$ (Line~20), causing the next iteration to greedily expand
$\vnew$.  Otherwise, $NULL$ is returned so as to do an exploration 
step.  Note that the approach is implemented with helpful 
branch-and-bound pruning after an initial solution is found, though
this is not reflected in the algorithmics.

\vnew\ is determined via an oracle function.  Using this oracle
function and a simple rewiring scheme is sufficient for showing
asymptotic optimality for $\drrtstar$ (see 
Section~\ref{sec:analysis}).  The oracle function $\oracle$ for a
two-robot case is illustrated in Figure \ref{fig:oracle}.  First, let
$\rho(Q,Q')$ be the ray from configuration $Q$ terminating at $Q'$.
Then, denote $\angle_{Q} (Q',Q'')$ as the minimum angle between
$\rho(Q,Q')$ and $\rho(Q,Q'')$.  When \qrand\ is drawn in $\ccross$,
its nearest neighbor \vnear\ in $\tree$ is found. Then, project the
points \qrand\ and \vnear\ into each robot space $\cspace_i$, i.e.,
ignore the configurations of other robots.

\begin{figure}[ht]
\centering
\includegraphics[height=1.45in]{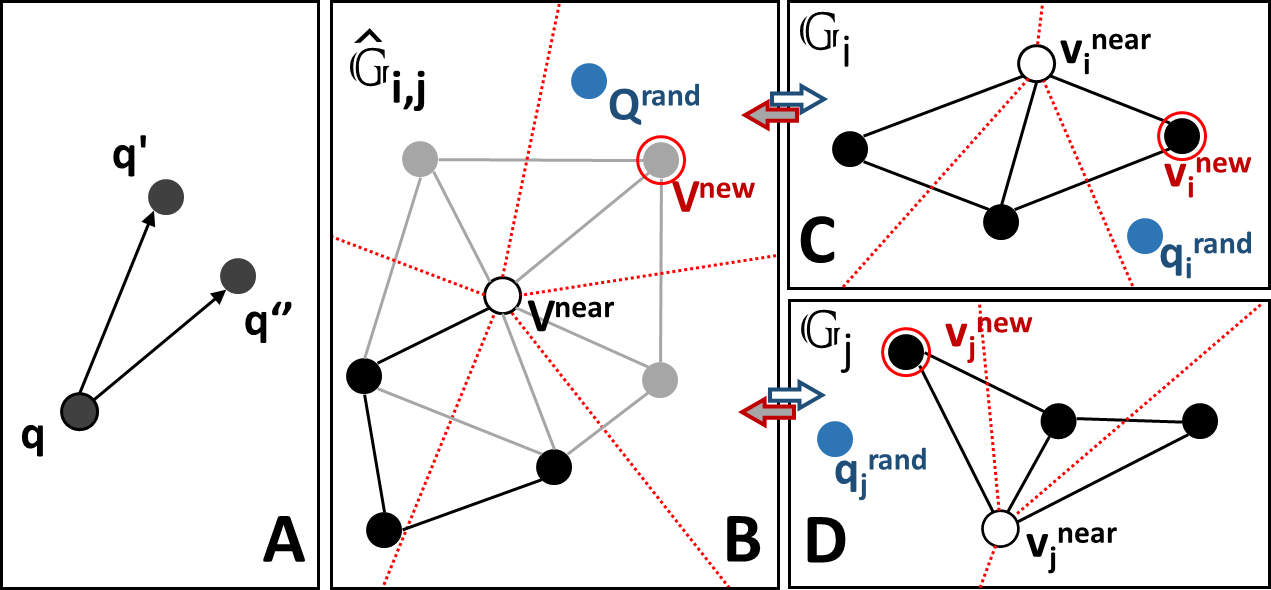}
\caption{(A) The method reasons over all neighbors $q'$ of $q$ so as
  to minimize the angle $\angle_{q}(q', q'')$. (B)
  $\oracle(\cdot,\cdot)$ finds graph vertex \vnew by minimizing angle
  $\angle_{\vnear}(\vnew,\qrand)$. (C,D) \vnear and \qrand are
  projected into each robot's $\cspace$-space so as to find nodes
  $v^{\textup{new}}_{\textup{i}}$ and $v^{\textup{new}}_{\textup{j}}$,
  respectively, which minimize angle $\angle_{
    v^{\textup{near}}_{\textup{i/j}}}
  (v^{\textup{new}}_{\textup{i/j}}, q^{\textup{rand}}_{\textup{i/j}}
  )$.}
\label{fig:oracle}
\end{figure}

The method separately searches the single-robot roadmaps to discover 
\vnew. Denote $\vnear= (v_1,\ldots,v_R), \qrand= (\tilde{q}_1,\ldots,
\tilde{q}_R)$.  For every robot $i$, let $N_i\subset \nodes_i$ be the 
neighborhood of $v_i \in \nodes_i$, and identify $v'_i = 
\argmin_{v \in N_i} \angle_{v_i} (q^{rand}_i, v)$.  The oracle 
function returns node $\vnew = (v'_1,\ldots,v'_R)$.

\begin{figure}[ht]
\centering
\includegraphics[height=1.1in]{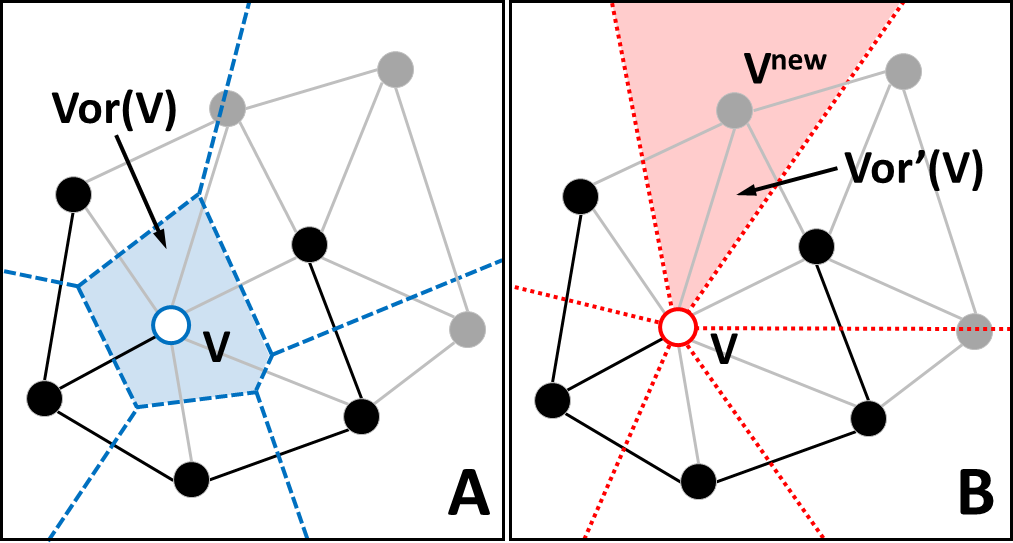}
\hspace{-0.2in}
\includegraphics[height=1.1in]{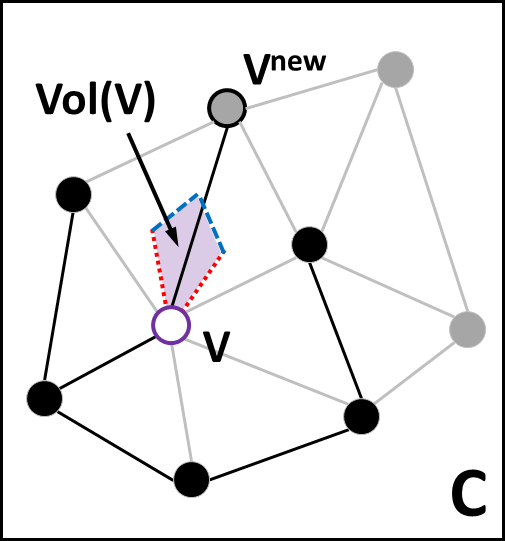}
\caption{(A) The Voronoi region $\textup{Vor}(V)$ 
of vertex $V$ is shown where if \qrand\ is drawn, vertex $V$ is 
selected for expansion. (B) When \qrand\ lies in the directional 
Voronoi region $\textup{Vor}'(V)$, the expand step expands to \vnew.  
(C) Thus, when \qrand\ is drawn within volume $\textup{Vol}(V) = 
\textup{Vor}(V) \cup \textup{Vor}'(V)$, the method will generate 
\vnew\ via $V$.}
\label{fig:voronoi}
\end{figure}

As in the standard $\rrt$ as well as in $\drrt$, the $\drrtstar$
approach has a Voronoi-bias property
~\cite{DBLP:conf/icra/LindemannL04}.  It is, however, slightly more
involved to observe as shown in Figure \ref{fig:voronoi}.  To generate
an edge $(V, V')$, random sample \qrand\ must be drawn within the
Voronoi cell of $V$, denoted $\textup{Vor}(V)$ (A) and in the general
direction of $V'$, denoted $\textup{Vor}'(V)$ (B).  The intersection
of these two volumes $\textup{Vol}(V) = \textup{Vor}(V) \cap
\textup{Vor}'(V)$ is the volume to be sampled generate \vnew\ via
\vnear.

%%% Local Variables:
%%% mode: latex
%%% TeX-master: "../isrr"
%%% End:

\section{Analysis}
\label{sec:analysis}

In this section, the theoretical properties of \drrtstar\ are 
examined, beginning with a study of the asymptotic convergence of the 
implicit roadmap $\mmgimp$ to containing a path in $\cfree$ whose cost 
converges to the optimum.  Then, it is shown \drrtstar\ eventually 
discovers the shortest path in $\mmgimp$, and that the combination of 
these two facts proves the asymptotic optimality of \drrtstar.

For simplicity, the analysis is restricted to the setting of robots
operating in Euclidean space, i.e. $\cspace_i$ is a $d$-dimensional
Euclidean hypercube $[0,1]^d$ for fixed $d \geq 2$. \footnote{For 
simplicity, it is assumed that all the robots have the same number  
of degrees of freedom $d$.}  Additionally, the analysis is restricted 
to the specific cost function of \emph{total distance}, i.e., 
$|\Sigma|:=\sum_{i=1}^R|\sigma_i|$. Discussion on lifting these 
restrictions is provided in Section \ref{sec:discuss}.

\subsection{Optimal Convergence of $\mmgimp$}

For each robot, an asymptotically optimal \prmstar\ roadmap 
$\mmgraph_i$ is constructed having $n$ samples and using a connection
radius $r(n)$ necessary for asymptotic convergence to the optimum
\cite{Karaman2011Sampling-based-}.  By the nature of sampling-based
algorithms, each graph cannot converge to the true optimum with finite
computation, as such a solution may have clearance of exactly $0$.  
Instead, this work focuses on the notion of a robust optimum
\footnote{Note that the given definition of robust optimum is similar 
    to that in previous work~\cite{Solovey_2016_rss}.}, showing that 
the tensor product roamdap $\mmgimp$ converges to this value.

\begin{definition}
  A trajectory $\Sigma:[0,1]\rightarrow \cfree$ is \emph{robust} if
  there exists a fixed $\delta>0$ such that for every
  $\tau\in [0,1],X\in \cinv$ it holds that
  $\|\Sigma(\tau)-X\|_2\geq \delta$, where $\|\cdot\|_2$ denotes the
  standard Euclidean distance.
\end{definition}

\begin{definition}
  A value $c > 0$ which denotes a path cost is robust if for every
  fixed $\epsilon > 0$ there exists a robust path $\Sigma$ such that
  $|\Sigma| \leq (1+\epsilon)c$. The \emph{robust optimum},
  denoted by $c^*$, is the infimum over all such values.
\end{definition}

For any fixed $n\in \mathbb{N}^+$, and a specific instance of
$\mmgimp$ constructed from $R$ roadmaps, having $n$ samples each,
denote by $\Sigma^{(n)}$ the shortest path from $S$ to $T$ 
over~$\mmgimp$.

\begin{definition}
  $\mmgimp$ is asymptotically optimal (AO) if for every fixed
  $\epsilon > 0$ it holds that
  $|\Sigma^{(n)}| \leq (1+\epsilon)c^*$ a.a.s.\footnote{Let
      $A_1,A_2,\ldots$ be random variables in some probability space
      and let $B$ be an event depending on~$A_n$. We say that $B$
      occurs \emph{asymptotically almost surely} (a.a.s.) if
      $\lim_{n\rightarrow \infty}\Pr[B(A_n)]=1$.}, where the
    probability is over all the instantiations of $\mmgimp$ with $n$
    samples for each PRM.
\end{definition}

Using this definition, the following theorem is proven.  Recall that 
$d$ denotes the dimension of a single-robot configuration space. 
 
\begin{theorem}
$\mmgimp$ is AO when $$r(n)\geq \radstar=(1+\eta)2 \left(\frac{1}{d} 
\right)^{\frac{1}{d}} \left( \frac{\log n}{n} \right)^{\frac{1}{d}},$$
where $\eta$ is any constant larger than $0$.
\label{thm:opt_graph}
\end{theorem}

{\bf Remark.} Note that \radstar\ was developed in
\cite[Theorem~4.1]{Pavone:2015fmt}, and guarantees AO of \prmstar for 
a single robot. The proof technique described in that work will be one
of the ingredients used to prove Theorem~\ref{thm:opt_graph}. 
\footnote{Note that \radstar\ can be refined to incorporate the 
proportion of $\cfree_i$, which would reduce this expression.} 

By definition of $c^*$, for any given $\epsilon>0$ there exists
robust trajectory $\Sigma:[0,1]\rightarrow \cfree$, and fixed
$\delta>0$, such that the cost of $\Sigma$ is at most 
$(1+1/2\cdot \epsilon)c^*$ and for every
$X\in \cinv, \tau\in [0,1]$ it holds that
$\|\Sigma(\tau)-X\|\geq \delta$.  Next, it is shown that $\mmgimp$
contains a trajectory $\Sigma^{(n)}$ such that
\begin{equation}
|\Sigma^{(n)}|\leq (1+o(1))\cdot |\Sigma|, \label{eq:eps_approx}
\end{equation}
a.a.s.. This immediately implies that
$|\Sigma^{(n)}| \leq (1+\epsilon)c^*$, which will finish the proof of
Theorem~\ref{thm:opt_graph}.

Thus, it remains to show that there exists a trajectory on $\mmgimp$ 
which satisfies Equation~\ref{eq:eps_approx} a.a.s..  As a first
step, it will be shown that the robustness of
$\Sigma = (\sigma_1,\ldots,\sigma_R)$ in the composite space implies
robustness in the single-robot setting, i.e., robustness along
$\sigma_i$.

For $\tau \in [0,1]$ define the forbidden space parameterized by
$\tau$ as \vspace{-0.1in} $$\cinv_i(\tau) = \cinv_i \cup 
\bigcup_{j=1, j \neq i}^R I_i^j( \sigma_j (\tau) ).$$ 

\begin{claim}
For every robot $i$, $\tau \in [0,1]$, and $q_i \in \cinv_i(\tau)$, 
$\| \sigma_i(\tau) - q_i \|_2 \geq \delta$.
\label{claim:robust}
\end{claim}
\begin{proof}
  Fix a robot $i$, and fix some $\tau \in [0,1]$ and a 
  configuration $q_i \in \cinv_i(\tau)$.  Next, define the
  following composite configuration
$$Q = (\sigma^1(\tau), \dots, q_i, \dots , \sigma^R(\tau)).$$
Note that it differs from $\Sigma(\tau)$ only in the $i$-th robot's
configuration. By the robustness of $\Sigma$ it follows that
\begin{align*}
\delta & \leq \| \Sigma(\tau) - Q \|_2\\ 
       & = \left( \| \sigma_i(\tau) - q_i \|_2^2 + \sum_{j=1,j \neq i}^R \| \sigma_j(\tau) -
           \sigma_j(\tau) \|_2^2 \right)^{\frac{1}{2}} \\
       &\leq \| \sigma_i(\tau) - q_i \|_2.
\end{align*}
\end{proof}

The result of claim \ref{claim:robust} is that the paths
$\sigma_1, \dots, \sigma_R$ are robust in the sense that there is
sufficient clearance for the individual robots to not collide with
each other given a fixed location of a single robot.  A Lemma is 
derived using proof techniques from the 
literature~\cite{Pavone:2015fmt}, and it implies every $\graph_i$ 
contains a single-robot path $\sigma_i^{(n)}$ that converges to 
$\sigma_i$
\begin{lemma}
For every robot $i$, $\graph_i$ constructed with $n$ samples and a
connection radius $r(n)\geq \radstar$ contains a 
path $\sigma_i^{(n)}$ with the following attributes a.a.s.: 
\begin{itemize}
\item[(i)] $\sigma_i^{(n)}(0) = s_i$, $\sigma_i^{(n)}(1) = t_i$; 
\item[(ii)] $|\sigma_i^{(n)}| \leq (1 + o(1)) |\sigma_i|$; 
\item[(iii)] $\forall q \in \textup{Im}(\sigma_i^{(n)})$, $\exists 
\tau \in [0,1]$ s.t. $\|q-\sigma_i(\tau)\|_2 \leq \radstar$.
\end{itemize}
\label{lem:prm}
\end{lemma}
\begin{proof}
  The first property (i) follows from the fact that $s_i,t_i$
  are directly added to $\graph_i$. The rest follows from the proof of
  Theorem~4.1 in~\cite{Pavone:2015fmt}, which is applicable here since
  $r(n)\geq \radstar$.
\end{proof}

Lemma~\ref{lem:prm} also implies that $\mmgimp$ contains a path in
$\cspace$, that represents robot-to-obstacle collision-free motions,
and minimizes the multi-robot metric cost. In particular, define
$\Sigma^{(n)}=(\sigma_1^{(n)},\ldots, \sigma_R^{(n)})$, where
$\sigma_i^{(n)}$ are obtained from Lemma~\ref{lem:prm}. Then
$$|\Sigma^{(n)}|=\sum_{i=1}^R|\sigma_i^{(n)}|\leq  (1 +
o(1))\sum_{i=1}^R|\sigma_i|\leq (1+o(1))|\Sigma|.$$
However, it is not clear whether this ensures the existence of a path
where robot-robot collisions are avoided.  That is, although
$\textup{Im}(\sigma^{(n)}_i)\subset \cfree_i$, it might be the case
that $\textup{Im}(\Sigma^{(n)})\cap \cinv \neq \emptyset$.  Next it is
shown that $\sigma_1^{(n)},\ldots, \sigma_R^{(n)}$ can be
reparametrized to induce a composite-space path whose image is fully
contained in $\cfree$, with length equivalent to $\Sigma^{(n)}$.

For each robot $i$, denote by $V_i=(v_i^1,\ldots,v_i^{\ell_i})$ the
chain of $\graph_i$ vertices traversed by $\sigma^{(n)}_i$. For every
$v_i^j\in V_i$ assign a timestamp $\tau_i^j$ of the closest
configuration along $\sigma^i$, i.e., \vspace{-0.1in}
$$\tau_i^j=\argmin_{\tau\in [0,1]}\|v_i^j-\sigma_i(\tau)\|_2.$$
Also, define $\T_i=(\tau_i^1,\ldots,\tau_i^{\ell_i})$ and denote by
$\T$ the ordered list of $\bigcup_{i=1}^R\T_i$, according to the
timestamp values. Now, for every $i$, define a global timestamp
function $T\!S_i:\T \rightarrow V_i$ which assigns to each global 
timestamp in $\T$ a single-robot configuration from $V_i$.  It thus
specifies in which vertex robot $i$ resides at time $\tau \in \T$.
For $\tau\in \T$, let $j$ be the largest index such that
$\tau_i^j \leq \tau$. Then simply assign $T\!S_i(\tau)= \tau_i^j$. 
From property (iii) in Lemma~\ref{lem:prm} and 
Claim~\ref{claim:robust} it follows that no robot-robot collisions are 
induced by the reparametrization, concluding the proof of
Theorem~\ref{thm:opt_graph}.

\subsection{Asymptotic Optimality of $\drrtstar$}

Finally, \drrtstar\ is shown to be AO.  Denote by $m$ the time budget 
in Algorithm~\ref{algo:drrtstar},  i.e., the number of iterations of 
the loop. Denote by $\Sigma^{(n,m)}$ the solution returned by 
$\drrtstar$ for $n$ and $m$.

\begin{theorem}
  If $r(n)>\radstar$ then for every fixed $\epsilon>0$ it holds that
  $$\lim_{n,m\rightarrow \infty}\Pr\left[|\Sigma^{(n,m)}|\leq
    (1+\epsilon)c^*\right]=1.$$
  \label{thm:ao_drrt}
\end{theorem}

Since $\mmgimp$ is AO (Theorem~\ref{thm:opt_graph}), it suffices to
show that for any fixed $n$, and a fixed instance of $\mmgimp$,
defined over $R$ PRMs with $n$ samples each, $\drrtstar$ eventually 
(as $m$ tends to infinity), finds the optimal trajectory over 
$\mmgimp$.  This can be shown using the properties of a Markov chain 
with absorbing states \cite[Theorem~11.3]{Snell2012:intro_prob}.
While a full proof is omitted here, the high-level idea is similar to
what is presented in previous work \cite[Theorem~3]{SoloveySH16:ijrr},
and expanded upon in Appendix~A.
By restricting the states of the Markov chain to being the graph
vertices along the optimal path, setting the target vertex to be an
absorbing vertex, and showing that the probability of transitioning
along any edge in this path is nonzero (i.e. the probability is
proportional to $\frac{\mu(\textup{Vol}(V_{k}))}{\mu(\cfree)} > 0$), 
then the probability that this process does not reach the target state 
along the optimal path converges to $0$ as the number of \drrtstar\
iterations tends to infinity.  The final step is to show that the
above statements hold when both $m$ and $n$ tend to $\infty$. A proof
for this phenomenon can be found
in~\cite[Theorem~6]{SoloveySH16:ijrr}.

%%% Local Variables:
%%% mode: latex
%%% TeX-master: "../isrr"
%%% End:

\section{Experimental Validation}
\label{sec:experiments}

This section provides an experimental evaluation of \drrtstar\ by
demonstrating practical convergence, scalability, and applicability to
dual-arm manipulation. The approach and alternatives are executed on a
cluster with Intel(R) Xeon(R) CPU E5-4650 @ 2.70GHz processors, and
128GB of RAM.  \footnote{Additional data are provided in Appendix
\ref{apx:Experiments}.}

\begin{wrapfigure}{r}{0.255\textwidth}
  \centering
  \includegraphics[width=0.250\textwidth]{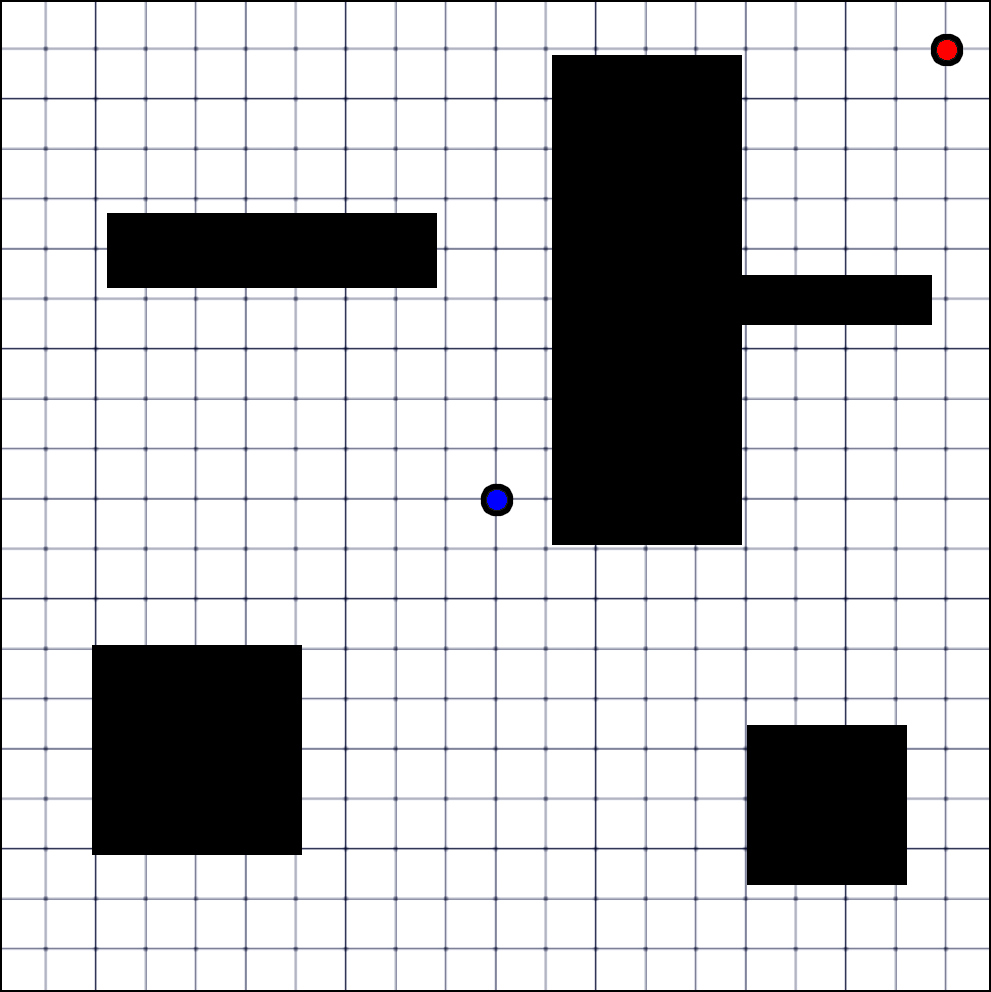}
  \caption{The 2D environment where the 2 disk robots 
    operate.}
  \label{fig:poly_enviro}
\end{wrapfigure}

\noindent \textbf{2 Disk Robots among 2D Polygons:} This
base-case test involves $ 2 $ disks ($\cspace_i := \reals^2$) of
radius $0.2$, in a $10.2 \times 10.2$ region, as in
Figure~\ref{fig:poly_enviro}. The disks have to swap positions between
$(0,0)$ and $(9,9)$. This is a setup where it is possible to compute
the explicit roadmap, which is not practical in more involved
scenarios. In particular, \drrtstar\ is tested against: a) running
$\astar$ on the implicit tensor roadmap $\mmgimp$ (referred to as
``Implicit \astar'') defined over the same individual roadmaps with
$N$ nodes each as those used by \drrtstar; and b) an explicitly
constructed \prmstar\ roadmap with $N^2$ nodes in the composite space.

Results are shown in Figure~\ref{fig:polygonal_benchmark}.  \drrtstar\
converges to the optimal path over $\mmgimp$, similar to the one
discovered by Implicit \astar, while quickly finding an initial
solution of high quality. Furthermore, the implicit tensor product
roadmap $\mmgimp$ is of comparable quality to the explicitly
constructed roadmap.

\begin{figure}[h]
\centering
\includegraphics[width=0.48\textwidth]{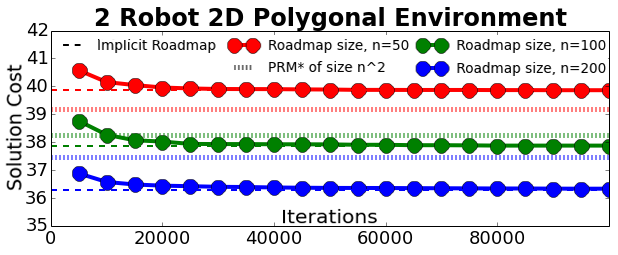}
\caption{ Average solution cost over iterations. Data averaged over 
$10$ roadmap pairs.  $\drrtstar$ (solid line) converges to the optimal 
path through $\mmgimp$ (dashed line).}
\label{fig:polygonal_benchmark}
\end{figure}

Table~\ref{tab:2_robot} presents running times.
\drrtstar\ and implicit \astar\ construct $2$ $N$-sized roadmaps 
(row~3), which are faster to construct than the \prmstar\ roadmap in
$\cspace$ (row~1).  \prmstar\ becomes very costly as $N$ increases.
For $N=500$, the explicit roadmap contains $250,000$ vertices, taking
$1.7$GB of RAM to store, which was the upper limit for the machine
used. When the roadmap can be constructed, it is quicker to query
(row~2). \drrtstar\ quickly returns an initial solution (row~5), and
converges within $5\%$ of the optimum length (row~6) well before
Implicit \astar\ returns a solution as $N$ increases (row~4). The next
benchmark further emphasizes this point.

\begin{table}[h]
\centering
\caption{Construction and query times (SECs) for 2 disk robots.}
\label{tab:2_robot}
\small
\begin{tabular}{|l|c|c|c|}
\hline
\multicolumn{1}{|r|}{\textbf{Number of nodes: $ N $ =}} & \textbf{50} & \textbf{100} & \textbf{200} \\ \hline
{$N^2$-PRM* construction}                     & 3.427        & 13.293        & 69.551        \\ \hline
{$N^2$-PRM* query}                            & 0.002       & 0.004        & 0.023        \\ \hline
{2 $N$-size PRM* construction}              & 0.1351        & 0.274        & 0.558       \\ \hline
{Implicit A* search over $\mmgimp$}                     & 0.684       & 2.497        & 10.184        \\ \hline
{\drrtstar\ over $\mmgimp$ (initial) }                    & 0.343       & 0.257        & 0.358        \\ \hline
{\drrtstar\ over $\mmgimp$ (converged) }                    & 3.497       & 4.418        & 5.429        \\ \hline
\end{tabular}
\end{table}

\textbf{Many Disk Robots among 2D Polygons:} In the same environment as
above, the number of robots $R$ is increased to evaluate scalability.
Each robot starts on the perimeter of the environment and is tasked
with reaching the opposite side. An $N=50$ roadmap is constructed for
every robot. It quickly becomes intractable to construct a \prmstar\
roadmap in the composite space of many robots.

\begin{figure}[h]
\centering
\includegraphics[width=0.475\textwidth]{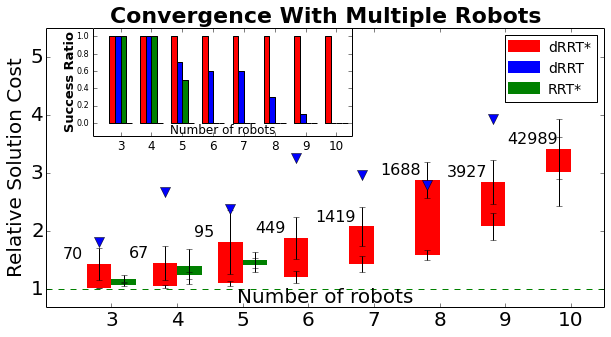}
\includegraphics[width=0.475\textwidth]{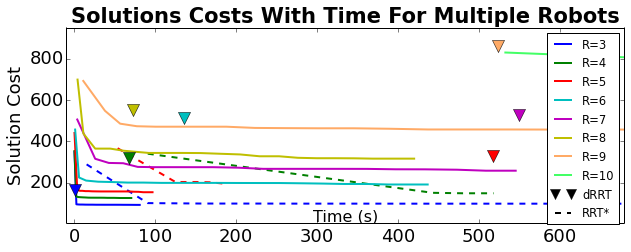}
\caption{  Data  averaged over $10$ runs.
(\textit{Top}): Relative solution cost and success ratio
of \drrtstar, \drrt\ and \rrtstar\ for increasing $R$. \drrtstar:
average iteration and variance for initial solution (top of box), and
solution cost and variance after $100,000$ iterations
(bottom). Similar results for \rrtstar.  Single data point for \drrt\
(no quality improvement after first solution). (\textit{Bottom}):
Solution costs over time.} 
\label{fig:scalability}
\end{figure}

Figure~\ref{fig:scalability} shows the inability of alternatives to
compete with \drrtstar\ in scalability. Solution costs are normalized
by an optimistic estimate of the path cost for each case, which is the
sum of the optimal solutions for each robot, disregarding robot-robot
interactions.  Implicit \astar\ fails to return solutions even for 3
robots. Directly executing \rrtstar\ in the composite space fails to
do so for $R \geq 6$.  The original \drrt\ method (without the
informed search component) starts suffering in success ratio for
$R \geq 5$ and returns worse solutions than \drrtstar.  The average
solution times for \drrt\ may decrease as $R$ increases but this is
due to the decreasing success ratio, i.e., \drrt\ begins to only
succeed at easy problems.

\begin{figure}[h]
\includegraphics[height=1.25in]{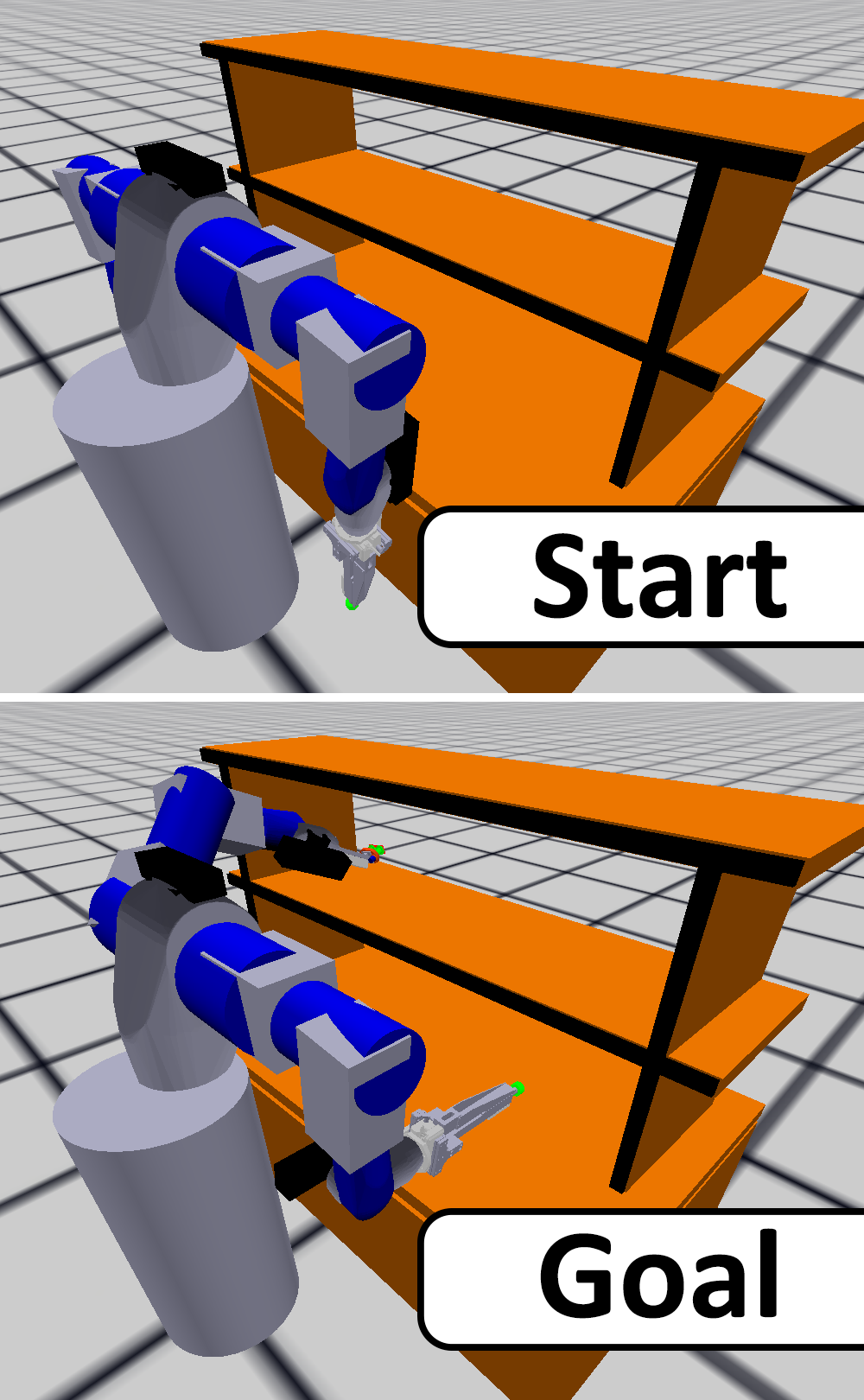}
\includegraphics[height=1.25in]{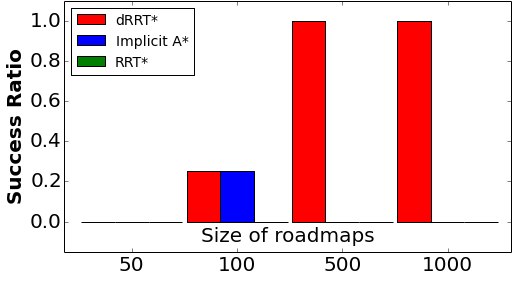}
\centering
\includegraphics[width=0.48\textwidth]{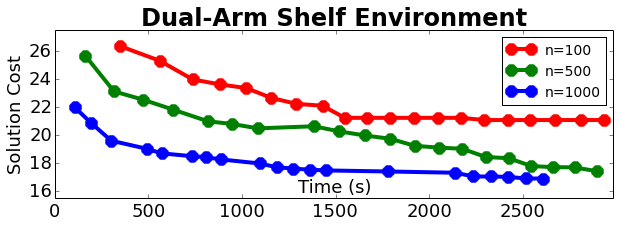}
\caption{
(\textit{Top}): \drrtstar\ is run for a dual-arm 
manipulator to go from its home position (above) to a reaching 
configuration (below) and achieves perfect success ratio as $n$ 
increases.
(\textit{Bottom}): \drrtstar\ solution quality over time.  Here,
larger roadmaps provide benefits in terms of running time and solution
quality.}
\label{fig:motoman_convergence}
\end{figure}

\textbf{Dual-arm manipulator:} This test shows the benefits
of \drrtstar\ when planing for two $7$-dimensional arms.
Figure~\ref{fig:motoman_convergence} shows that
\rrtstar\ fails to return solutions within $100K$
iterations. Using small roadmaps is also insufficient for this
problem.  Both \drrtstar\ and Implicit \astar\ require larger roadmaps
to begin succeeding. But with $N \geq 500$, Implicit \astar\ always
fails, while \drrtstar\ maintains a $100\%$ success ratio. As
expected, roadmaps of increasing size result in higher quality
path. The informed nature of \drrtstar\ also allows to find initial
solutions fast, which together with the branch-and-bound primitive
allows for good convergence.

%Future work will be directed toward using \drrtstar\ when a kinematic
%linkage between parts of a system is unavoidable.

\section{Discussion}
\label{sec:discuss}

This work studies the asymptotic optimality of sampling-based
multi-robot planning over implicit structures.  The objective is to
efficiently search the composite space of multi-robot systems while
also achieving formal guarantees. This requires construction of
asymptotically optimal individual robot roadmaps and appropriately
searching their tensor product. Performance is further improved by
informed search in the composite space.

These results may extend to more complex settings involving
kinodynamic constraints by relying on recent
work~\cite{ELM15a,ELM15b}.  Furthermore, the analysis may be valid for
different cost functions other than total distance. The tools
presented here can prove useful in the context of simultaneous task
and motion planning for multiple robots, especially
manipulators \cite{Dobson:2015_MAM}.

% \section{Benchmarking the AO Implicit Graph}
% \input{0X_benchmarks}

% \section{Efficient Implicit Graph Definition}
% \input{0Y_rss_results}

% \section{Application to Multi-Arm Manipulation}
% \input{0Z_application}

{\small
\bibliographystyle{abbrv}
\bibliography{manip}}

\appendix
\subsection{Proof of Convergence}
\label{apx:conv_proof}

This appendix examines the result of Theorem~\ref{thm:ao_drrt}, and
formally proves the convergence of the \drrtstar\ tree toward
containing all optimal paths.

\begin{lemma}[Optimal Tree Convergence of \drrtstar]
\label{lem:tree_conv}
Consider an arbitrary optimal path $\pi^*$ originating from $v_0$ and 
ending at $v_{t}$, then let $O^{(m)}_k$ be the event such that after 
$m$ iterations of \drrtstar, the search tree $\tree$ contains the 
optimal path up to segment $k$.  Then, $$ \liminf_{m \to \infty} \pr 
\big( O^{(m)}_t \big) = 1.$$
\end{lemma}

% = = = = = = = = = = = = 
%  MCMC Proof
% = = = = = = = = = = = =

{\bf Proof.} This property will be proven using a theorem from Markov
chain literature \cite[Theorem~11.3]{Snell2012:intro_prob}. 
Specifically, the properties of absorbing Markov chains can be 
exploited to show that $\drrtstar$ will eventually contain the optimal
path over $\mmgimp$ for a given query.  An absorbing Markov chain is
one such that there is some subset of states in which the transition
matrix only allows that state to transition to itself.

The proof follows by showing that the $\drrtstar$ method can be 
described as an absorbing Markov chain, where the target state of a
query is represented as an absorbing state in a Markov chain, 
re-stated here.

\begin{theorem}[Thm 11.3 in Grinstead \& Snell]
\label{thm:grinstead}
In an absorbing Markov chain, the probability that the process will be 
absorbed is 1 (i.e., $Q(m) \to 0$ as $n \to \infty$), where $Q(m)$ is
the transition submatrix for all non-absorbing states.
\end{theorem}

There are two steps to using this proof.  First, that the $\drrtstar$
search can be cast as an absorbing Markov chain, and second, that the
probability of transition for from each state to the next in this 
chain is nonzero (i.e. that each state can eventually be connected to
the target).

For query $(S, T)$, let the sequence $V = \{ v_1, v_2, \dots, 
v_{\textup{t}}\}$ of length $t$ represent the vertices of $\mmgimp$ 
corresponding to the optimal path through the graph which connects 
these points, where $v_{\textup{t}}$ corresponds to the target vertex,
and furthermore, let $v_{\textup{t}}$ be an absorbing state.  
Theorem~\ref{thm:grinstead} operates under the assumption that each 
vertex $v_{\textup{i}}$ is connected to an absorbing state 
($v_{\textup{t}}$ in this case).

Then, let the transition probability for each state have two values, 
one for each state transitioning to itself, which corresponds to the
$\drrtstar$ search expanding along some other arbitrary path.  The 
other value is a transition probability from $v_{\textup{i}}$ to 
$v_{\textup{i}+1}$.  This corresponds to the method sampling within
the volume $\textup{Vol}(v_{\textup{i}})$.

Then, as the second step, it must be shown that this volume has a 
positive probability of being sampled in each iteration.  It is 
sufficient then to argue that $\frac{\mu(\textup{Vol}
(s_{\textup{i}}))} {\mu(\cfree)} > 0$.  Fortunately, for any finite 
$n$, previous work has already shown that this is the case given 
general position assumptions \cite[Lemma~2]{SoloveySH16:ijrr}.

Given these results, the $\drrtstar$ is cast as an absorbing Markov
chain which satisfies the assumptions of \ref{thm:grinstead}, and 
therefore, the matrix $Q(m) \to 0$.  This implies that the optimal
path to the goal has been expanded in the tree, and therefore 
$ \liminf_{m \to \infty} \pr \big( O^{(m)}_t \big) = 1.$ \qed

\subsection{More Experimental Data}
\label{apx:Experiments}

This appendix presents additional experimental data omitted from 
Section~\ref{sec:experiments}.

%%% First
\subsubsection{2-Robot Benchmark}

\begin{figure}[ht]
    \centering
    \includegraphics[width=0.43\textwidth]{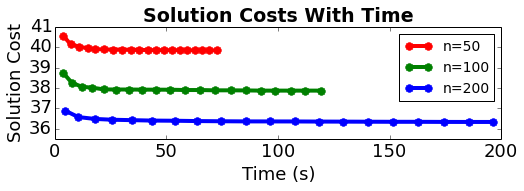}
    \caption{2-Robot convergence data over time.}
    \label{fig:2r_time}
\end{figure}

For the two-robot benchmark, additional data is presented in 
Figure~\ref{fig:2r_time}.  Here, the data presented in 
Figure~\ref{fig:polygonal_benchmark} is shown again over time instead 
of over iterations.

\begin{figure}[ht]
    \centering
    \includegraphics[width=0.43\textwidth]{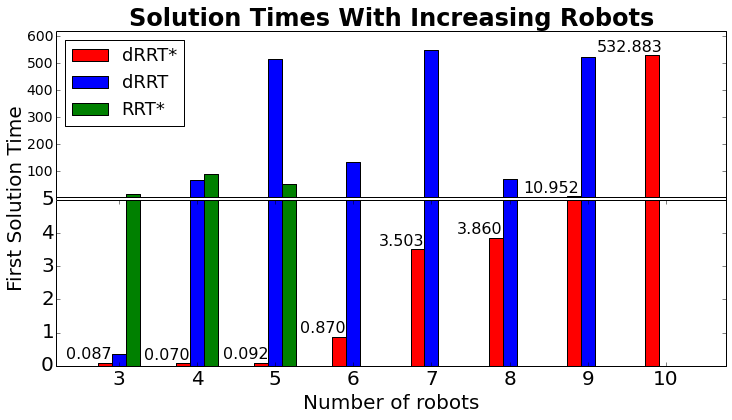}
    \caption{$R$-robot solution times for varying $R$.}
    \label{fig:rr_time}
\end{figure}

%%% Second
\subsubsection{R-Robot Benchmark}

For the $R$-robot benchmark, additional data is presented in 
Figure~\ref{fig:rr_time}, showing query resolution times for the 
various methods.

\begin{figure}[ht]
    \centering
    \includegraphics[width=0.43\textwidth]{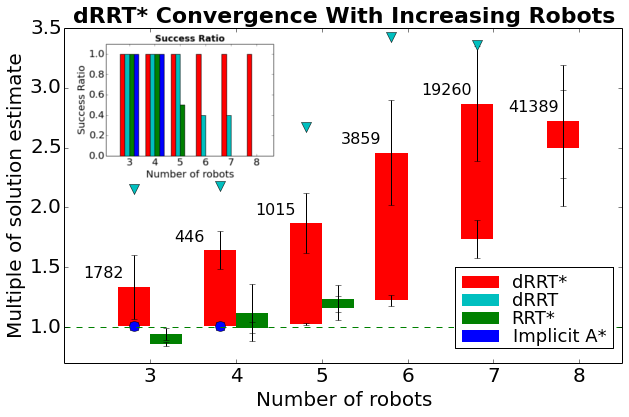}
    \includegraphics[width=0.43\textwidth]{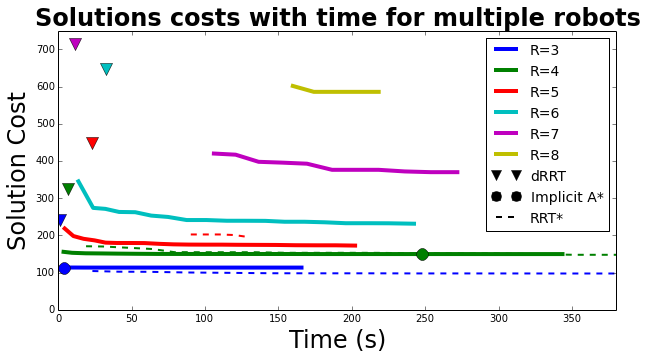}
    \caption{(\textit{Top}): Convergence rate and success ratio over 
    the minimal $9$-node roadmap (\textit{Bottom}): Solution cost over
    time when using the minimal roadmap.}
    \label{fig:synth_results}
\end{figure}

To emphasize the lack of scalability for alternate methods, additional
experiments were run in this setup using a minimal roadmap.  The tests
use a $9$-node roadmap for each robot as illustrated in 
Figure~\ref{fig:nine_grid}.  Each roadmap is constructed with slight
perturbations to the nodes within the shaded regions indicated in the
figure.

\begin{wrapfigure}{r}{0.168\textwidth}
  \centering
  \includegraphics[width=0.166\textwidth]{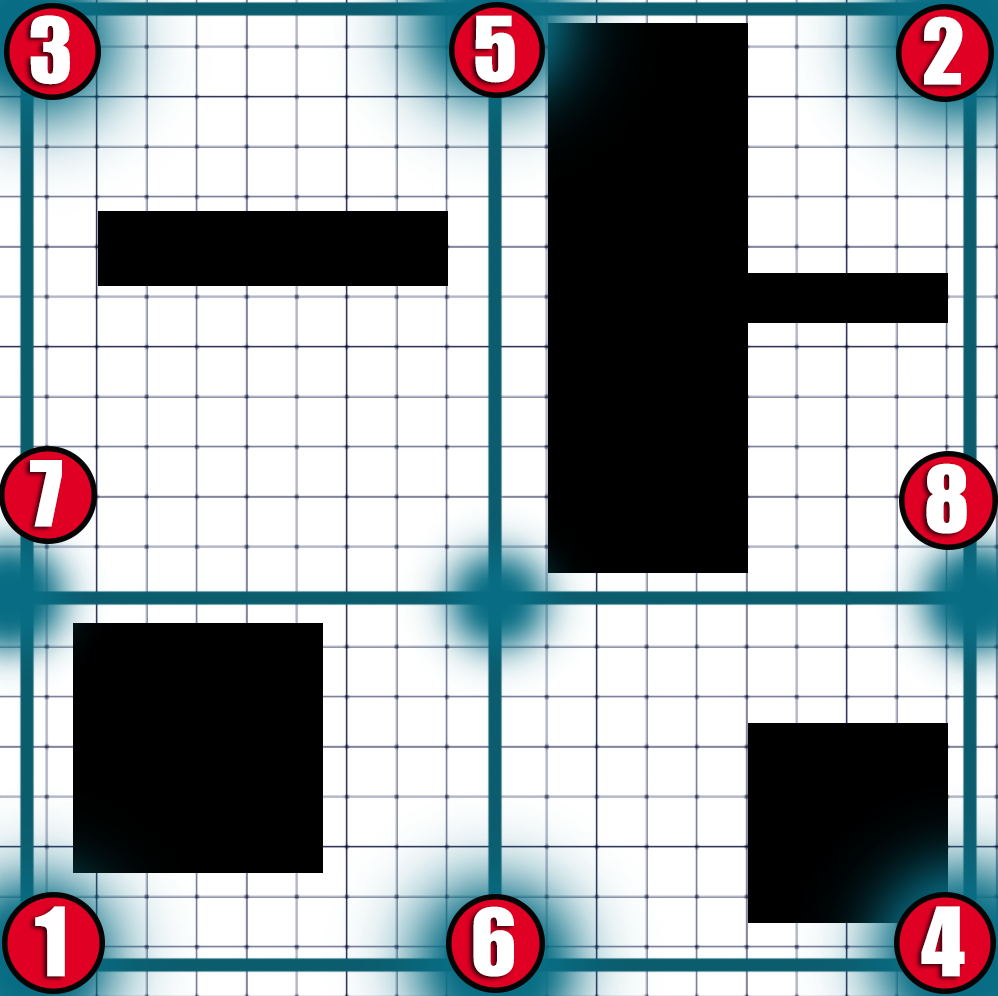}
  \caption{Minimal graph for the $R$-robot case.}
  \label{fig:nine_grid}
\end{wrapfigure}

The data for this modified benchmark (shown in 
Figure~\ref{fig:synth_results}) indicate that even using a very small
roadmap does not allow alternate methods to scale.  While the methods
scale better, Implict \astar\ does time out for $R \geq 5$, and 
\rrtstar\ times out for $R \geq 6$.

\subsubsection{Manipulator Benchmark}

For the dual-arm manipulator benchmark, additional data is presented 
in Figure~\ref{fig:moto_iter}.  Here, the data of 
Figure~\ref{fig:motoman_convergence} is shown over iterations 
instead of over time.

\begin{figure}[h]
    \centering
    \includegraphics[width=0.43\textwidth]{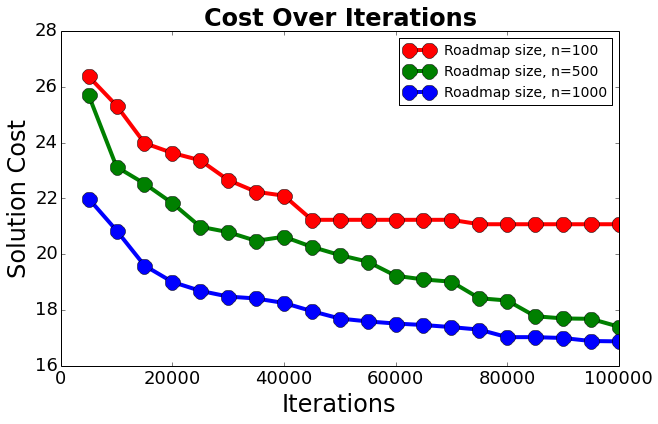}
    \caption{Motoman benchmark solution quality over iterations.}
    \label{fig:moto_iter}
\end{figure}

\end{document}